\documentclass[sn-mathphys-num]{sn-jnl}

\usepackage{amsmath,amssymb,amsthm}
\usepackage{graphicx}
\usepackage{hyperref}
\usepackage{tikz}
\usetikzlibrary{cd}

\theoremstyle{thmstyleone}
\newtheorem{theorem}{Theorem}[section]
\newtheorem{proposition}[theorem]{Proposition}
\newtheorem{lemma}[theorem]{Lemma}
\newtheorem{corollary}[theorem]{Corollary}
\theoremstyle{thmstyletwo}

\newtheorem{remark}{Remark}[section]
\theoremstyle{thmstylethree}
\newtheorem{definition}{Definition}[section]

\begin{document}


\title[Line Knowledge Digraphs and Sheaf Semantics]{From Line Knowledge Digraphs to Sheaf Semantics:
A Categorical Framework for Knowledge Graphs}


\author*{\fnm{Moses} \sur{Boudourides}}\email{Moses.Boudourides@northwestern.edu}

\affil{\orgdiv{School of Professional Studies},
          \orgname{Northwestern University},
          \orgaddress{\city{Evanston}, \state{IL}, \country{USA}}}


\abstract{This paper proposes a categorical framework for knowledge graphs
linking combinatorial graph structure with topos-theoretic semantics.
Knowledge graphs are represented as labelled directed multigraphs and
analysed through incidence matrices and line knowledge digraph
constructions. The graph induces a free category whose morphisms
correspond to relational paths. To model context-dependent meaning, a
Grothendieck topology is defined on the free category generated by the graph
leading to a topos of sheaves that supports local-to-global semantic reasoning. The framework connects graph-theoretic structure, categorical composition, and sheaf semantics in a unified mathematical model for contextual relational reasoning.}


\keywords{Knowledge graphs, Category theory, Topos theory,
          Sheaf semantics, Graph theory, Grothendieck topology}



\maketitle

\section{Introduction}\label{sec:intro}

Knowledge graphs are widely used structures for representing relational
data. They encode entities and relations through labelled triples and are
fundamental to semantic web technologies, digital humanities, and machine
learning~\cite{hogan2021kg,nickel2016review}.

Recent work in digital humanities, computational social science, and computational cultural analysis
has emphasised the importance of relational data representations for
studying social practices, artistic and cultural production, and interpretation. Approaches such as
distant reading analyse cultural corpora through structural patterns
rather than through individual artefacts \cite{moretti2013distantreading}.
Similarly, cultural analytics investigates large collections of visual
and cultural objects using computational and network-based methods
\cite{manovich2020culturalanalytics}. Within these contexts, knowledge
graphs provide a natural formalism for representing relations among
artworks, artists, historical periods, and iconographic motifs.

While their combinatorial structure is well understood, the semantic
structure of knowledge graphs remains less formally characterised. In
particular, standard graph-database models do not provide a principled
account of context-dependent or multi-perspective interpretations of the
same underlying facts.

This paper develops a categorical framework for analysing knowledge graphs.
Our approach connects combinatorial constructions on triples with
categorical and logical structures derived from the graph. In particular,
we study incidence-based line digraph constructions on triples, interpret
knowledge graphs as generating free categories, and investigate how these
categorical structures can be enriched through sheaf-theoretic and
topos-theoretic methods. The resulting framework links graph-theoretic
incidence constructions with categorical composition and local-to-global
semantic reasoning. 

Beyond establishing the basic categorical framework, the paper also
introduces a structural result relating different semantic
interpretations of a knowledge graph. By equipping the free category
$\mathcal{C}(K)$ with two distinct Grothendieck topologies—a
path-covering topology and an atomic topology—we obtain two associated
topoi of sheaves. We prove that the identity functor on
$\mathcal{C}(K)$ induces an essential geometric morphism between these
topoi. This result shows that different interpretive logics on the same
knowledge graph can be related through canonical categorical
constructions.

\subsection{Categorical viewpoint}

Knowledge graphs admit a natural categorical interpretation through
the free category generated by the graph.

Given a knowledge graph $\mathcal{K}=(E,P,T)$ (see \ref{def:kg}), one can construct the free
category $\mathbf{C}(\mathcal{K})$ (see \ref{sec:freecat}), whose objects are the entities in $E$
and whose generating morphisms correspond to triples
$(h,p,t)\in T$ interpreted as arrows
\[
h \xrightarrow{p} t .
\]
Morphisms in $\mathbf{C}(\mathcal{K})$ are finite paths of triples and
composition is given by concatenation of paths
\cite{maclane1998cwm,awodey2010,borceux1994handbook}.

Within this categorical framework the head and tail maps
\[
h,t : T \to E
\]
can be interpreted as domain and codomain assignments for the generating
morphisms. The incidence matrices introduced in the next section encode
these maps combinatorially and lead naturally to line knowledge digraph
constructions that capture structural relations between triples.

\subsection{Topos-theoretic perspective}

The categorical structure generated by a knowledge graph can be enriched
by equipping the free category $\mathbf{C}(\mathcal{K})$ with a suitable
Grothendieck topology (see \ref{sec:sites}). The resulting category of sheaves (see \ref{sec:topos})
\[
\mathbf{Sh}(\mathbf{C}(\mathcal{K}),J)
\]
forms a Grothendieck topos. Topos theory provides a logical environment
in which relational structures can be interpreted through local-to-global
principles~\cite{maclane1992sheaves,johnstone2002elephant}. In this
setting, compatible local pieces of relational information can be glued
together into global interpretations.

\subsection{Overview of contributions}

This paper develops a categorical framework for analysing knowledge
graphs that connects combinatorial graph structure with categorical
and topos-theoretic semantics. Incidence-based constructions on
triples lead to line knowledge digraphs capturing structural relations
between triples. The knowledge graph is then interpreted as generating
a free category whose morphisms correspond to relational paths.
Equipping this category with a Grothendieck topology yields a sheaf
topos providing a formal environment for contextual semantic
interpretation.

Sheaves and Topos Construction

\section{Knowledge Graphs as Edge-Labelled Digraphs}\label{sec:kg}

From a combinatorial viewpoint, knowledge graphs generalise directed
graphs by allowing edges to carry semantic labels representing
relations. This perspective connects knowledge graphs with classical
graph-theoretic structures while preserving the relational semantics
encoded by predicates. Such representations make it possible to apply
matrix-based and structural methods from graph theory to the analysis
of relational data \cite{harary1969}.

Thus, knowledge graphs can be represented as directed edge-labelled
multigraphs. This combinatorial representation provides the basis
for matrix-based and categorical analysis.

Higher-order generalisations of graph structures are also frequently
considered in the analysis of relational data. In particular,
simplicial complexes and hypergraphs allow interactions among more
than two entities to be represented within a single relational object.
Such structures are widely used in network science and computational
topology to capture multi-way relationships that cannot be reduced to
binary edges \cite{ghrist2014,battiston2020}.
In categorical terms, a $2$-simplex representing a ternary relation
can be interpreted as a $2$-cell in a higher-dimensional categorical
structure encoding multi-entity interactions.

Although the present work focuses on the triple-based structure of
knowledge graphs, the categorical framework developed below naturally
extends to such higher-dimensional relational models. In particular,
hypergraph and simplicial constructions may be interpreted as
generating higher-order categorical structures whose morphisms encode
multi-entity interactions. The incidence constructions and sheaf
semantics introduced here therefore provide a foundation that can be
extended to richer relational settings.

\begin{definition}[Knowledge graph]\label{def:kg}
A \emph{knowledge graph} is a triple
\[
\mathcal{K}=(E,P,T)
\]
where
\begin{itemize}
\item $E$ is a set of entities,
\item $P$ is a set of predicates,
\item $T \subseteq E \times P \times E$ is a set of triples.
\end{itemize}
A triple $\tau=(h,p,t)\in T$ represents a directed labelled edge
\[
h \xrightarrow{p} t .
\]
\end{definition}

Fix orderings
\[
E=\{e_1,\dots,e_n\}, \qquad
T=\{\tau_1,\dots,\tau_m\}.
\]

\subsection{Incidence matrices}

Incidence matrices provide a compact algebraic representation of the
interaction between vertices and edges in a graph. In the present
setting they encode the relationships between entities and triples,
allowing combinatorial properties of the knowledge graph to be studied
through linear-algebraic operations on matrices
\cite{harary1969,diestel2017}.

\begin{definition}[Head and tail incidence matrices]
The \emph{head-incidence matrix}
\[
H^{(h)}\in\{0,1\}^{n\times m}
\]
is defined by
\[
H^{(h)}_{ij}=
\begin{cases}
1 & \text{if } e_i=h(\tau_j),\\
0 & \text{otherwise}.
\end{cases}
\]

The \emph{tail-incidence matrix}
\[
H^{(t)}\in\{0,1\}^{n\times m}
\]
is defined by
\[
H^{(t)}_{ij}=
\begin{cases}
1 & \text{if } e_i=t(\tau_j),\\
0 & \text{otherwise}.
\end{cases}
\]
\end{definition}

Each column of $H^{(h)}$ contains exactly one non-zero entry corresponding
to the head of the triple, while each column of $H^{(t)}$ contains exactly
one non-zero entry corresponding to the tail.

\begin{proposition}\label{prop:incidence-symmetry}
Let $\mathcal{K} = (E, P, T)$ be a knowledge graph, and let $H^{(h)}$ and $H^{(t)}$ denote the head and tail incidence matrices of $\mathcal{K}$, respectively, as introduced in Section~\ref{sec:kg}. Then the matrices

\[
(H^{(h)})^\top H^{(h)}
\quad\text{and}\quad
(H^{(t)})^\top H^{(t)}
\]

\noindent are symmetric and encode shared-head and shared-tail relations between
triples.
\end{proposition}

\begin{proof}
The $(i,j)$ entry of $(H^{(h)})^\top H^{(h)}$ equals

\[
\sum_{k=1}^n H^{(h)}_{ki} H^{(h)}_{kj}.
\]

This counts entities that appear as heads of both triples $\tau_i$ and
$\tau_j$. Since each triple has exactly one head, the value is $1$ when
the heads coincide and $0$ otherwise. Symmetry follows immediately.
The argument for $(H^{(t)})^\top H^{(t)}$ is identical.
\end{proof}

\begin{lemma}\label{lem:line-operator}
Let $\mathcal{K}=(E,P,T)$ be a knowledge graph with head and tail
incidence matrices $H^{(h)}$ and $H^{(t)}$. Define

\[
A_{\mathrm{out}}=(H^{(h)})^\top H^{(h)}-I_m,
\qquad
A_{\mathrm{in}}=(H^{(t)})^\top H^{(t)}-I_m
\]

where $I_m$ is the $m\times m$ identity matrix.

Then $A_{\mathrm{out}}$ and $A_{\mathrm{in}}$ are adjacency matrices of
the out-line and in-line knowledge digraphs respectively.
\end{lemma}

\begin{proof}
The entry $(i,j)$ of $(H^{(h)})^\top H^{(h)}$ equals

\[
\sum_{k=1}^{n} H^{(h)}_{ki}H^{(h)}_{kj}.
\]

This quantity equals $1$ precisely when the triples $\tau_i$ and
$\tau_j$ share the same head entity and $0$ otherwise. The diagonal
entries equal $1$ because every triple shares its head with itself.
Subtracting the identity matrix therefore removes these self-incidences,
yielding

\[
A_{\mathrm{out}}(i,j)=
\begin{cases}
1 & \text{if } i\neq j \text{ and } h(\tau_i)=h(\tau_j),\\
0 & \text{otherwise}.
\end{cases}
\]

This coincides with the adjacency definition of the out-line knowledge
digraph. The same argument applied to $H^{(t)}$ yields the matrix
$A_{\mathrm{in}}$ for the in-line digraph.
\end{proof}

\begin{proposition}\label{prop:spectrum}
Let $\mathcal{K}=(E,P,T)$ be a knowledge graph and let

\[
A_{\mathrm{out}}=(H^{(h)})^\top H^{(h)}-I_m
\]

\noindent be the adjacency matrix of the out-line knowledge digraph.
Let $m_e$ denote the number of triples whose head entity is $e$.

Then the spectrum of $A_{\mathrm{out}}$ consists of the eigenvalues

\[
\{\,m_e-1 : e\in E\,\}
\]

\noindent together with $-1$ with multiplicity

\[
m-\lvert E_h\rvert
\]

where $E_h\subseteq E$ is the set of entities that occur as heads of
triples.
\end{proposition}

\begin{proof}
By Proposition~\ref{prop:component-partition}, the out-line digraph
decomposes as a disjoint union of complete directed graphs on the sets

\[
T_e=\{\tau\in T:h(\tau)=e\}.
\]

If $|T_e|=m_e$, the adjacency matrix of this component is

\[
J_{m_e}-I_{m_e}
\]

\noindent where $J_{m_e}$ is the all-ones matrix.

The eigenvalues of $J_{m_e}-I_{m_e}$ are well known (see \cite[Ch.~1]{brouwer2012}): $m_e-1$ with
multiplicity $1$ and $-1$ with multiplicity $m_e-1$.
Since the full adjacency matrix is block diagonal over the components,
its spectrum is the union of the spectra of these blocks.
\end{proof}

\section{Line Knowledge Digraphs}\label{sec:line}

The incidence representation introduced above allows us to construct
derived graphs whose vertices correspond to triples and whose adjacency
records shared incidence relations.

\begin{definition}[Out-line and in-line knowledge digraphs]
Let $\mathcal{K}=(E,P,T)$ be a knowledge graph with
$T=\{\tau_1,\dots,\tau_m\}$.

The \emph{out-line knowledge digraph}
$L_{\mathrm{out}}(\mathcal{K})$ has vertex set $T$ and a directed edge
$\tau_i\to\tau_j$ whenever $h(\tau_i)=h(\tau_j)$ and $i\neq j$.

The \emph{in-line knowledge digraph}
$L_{\mathrm{in}}(\mathcal{K})$ has vertex set $T$ and a directed edge
$\tau_i\to\tau_j$ whenever $t(\tau_i)=t(\tau_j)$ and $i\neq j$.
\end{definition}

\subsection{Incidence interpretation}

The construction of line digraphs has classical precedents in graph
theory, where vertices represent edges of the original graph and
adjacency reflects incidence relations between edges
\cite{beineke1970,rooij1965}. The present formulation adapts this idea to the
setting of knowledge graphs by treating triples as vertices and
identifying adjacency through shared head or tail entities.

Let $H^{(h)}$ and $H^{(t)}$ denote the head and tail incidence matrices.
Define
\[
M_{\mathrm{out}}=(H^{(h)})^\top H^{(h)}, \qquad
M_{\mathrm{in}}=(H^{(t)})^\top H^{(t)}.
\]

The matrices $M_{\mathrm{out}}$ and $M_{\mathrm{in}}$ therefore encode
the shared-head and shared-tail relations between triples. Similar
matrix formulations appear in classical treatments of line graphs,
where adjacency relations between edges are obtained through incidence
matrices of the underlying graph \cite{harary1969,beineke1970}.

\begin{proposition}
For $i\neq j$ we have
\[
\tau_i\to\tau_j \text{ in } L_{\mathrm{out}}(\mathcal{K})
\quad\text{iff}\quad
M_{\mathrm{out}}(i,j)>0,
\]
and
\[
\tau_i\to\tau_j \text{ in } L_{\mathrm{in}}(\mathcal{K})
\quad\text{iff}\quad
M_{\mathrm{in}}(i,j)>0.
\]
\end{proposition}

\begin{proof}
By definition,
\[
M_{\mathrm{out}}(i,j)=\sum_{k=1}^n
H^{(h)}_{ki}H^{(h)}_{kj}.
\]
This sum counts entities that occur as heads of both triples.
Since each triple has a unique head, the value is $1$ when the
triples share a head and $0$ otherwise.
The argument for $M_{\mathrm{in}}$ is identical using
$H^{(t)}$.
\end{proof}

\subsection{Structural decomposition}

Define relations on triples by
\[
\tau_i\sim_h\tau_j \iff h(\tau_i)=h(\tau_j),\qquad
\tau_i\sim_t\tau_j \iff t(\tau_i)=t(\tau_j).
\]

\begin{theorem}\label{thm:scc}
The strongly connected components of
$L_{\mathrm{out}}(\mathcal{K})$
correspond exactly to the equivalence classes of $\sim_h$,
and the strongly connected components of
$L_{\mathrm{in}}(\mathcal{K})$
correspond exactly to the equivalence classes of $\sim_t$.
\end{theorem}

\begin{proof}
Let $C_e=\{\tau\in T:h(\tau)=e\}$.
For any $\tau_i,\tau_j\in C_e$ we have both edges
$i\to j$ and $j\to i$, so the induced subgraph is
complete. If two triples have different heads,
no edge can exist between them. Thus each component
coincides with one such class. The same argument
applies to the in-line graph using tails.
\end{proof}

\begin{proposition}\label{prop:component-partition}
Let $\mathcal{K}=(E,P,T)$ be a knowledge graph. The vertex set $T$
of the out-line digraph $L_{\mathrm{out}}(\mathcal{K})$ admits the
partition
\[
T=\bigsqcup_{e\in E}T_e,
\qquad
T_e=\{\tau\in T:h(\tau)=e\}.
\]
Each set $T_e$ induces a complete directed subgraph of
$L_{\mathrm{out}}(\mathcal{K})$.
\end{proposition}

\begin{proof}
This follows immediately from Theorem~\ref{thm:scc}.
\end{proof}

\begin{corollary}\label{cor:component-degree}
Let $\tau\in T$ be a triple whose head entity occurs in $m$ triples of
$\mathcal{K}$. Then the corresponding vertex in
$L_{\mathrm{out}}(\mathcal{K})$ has out-degree and in-degree equal to
$m-1$.
\end{corollary}

\begin{proof}
All triples sharing the same head entity form a complete directed
subgraph in $L_{\mathrm{out}}(\mathcal{K})$. If $m$ triples share this
head, each triple is adjacent to the other $m-1$ triples, giving the
stated degree.
\end{proof}

\begin{remark}[Categorical interpretation of line digraphs]
The decomposition described above has a natural categorical
interpretation in terms of the free category $\mathcal{C}(K)$ generated
by the knowledge graph. The triples of $K$ correspond to generating
morphisms of $\mathcal{C}(K)$, and the relations $\sim_h$ and $\sim_t$
group together morphisms that share the same domain or codomain.
Consequently, the strongly connected components of the line knowledge
digraphs reflect the domain and codomain fibres of the generating
morphisms in $\mathcal{C}(K)$. In this way the combinatorial structure
of the line digraphs mirrors the categorical organisation of morphisms
in the free category.
\end{remark}

\begin{proposition}\label{prop:rank-incidence}
Let $\mathcal{K}=(E,P,T)$ be a knowledge graph with $|E|=n$ entities
and $|T|=m$ triples, and let $H^{(h)},H^{(t)}\in\{0,1\}^{n\times m}$
be the head and tail incidence matrices.

Then each column of $H^{(h)}$ and $H^{(t)}$ contains exactly one
non-zero entry, and consequently

\[
\operatorname{rank}(H^{(h)}) \le n,
\qquad
\operatorname{rank}(H^{(t)}) \le n.
\]

Moreover, if every entity appears as the head of at least one triple,
then

\[
\operatorname{rank}(H^{(h)}) = n .
\]
\end{proposition}

\begin{proof}
Each column of $H^{(h)}$ contains a single entry equal to $1$
corresponding to the head of the triple and zeros elsewhere.
Hence the column space of $H^{(h)}$ lies in $\mathbb{R}^n$, implying
$\operatorname{rank}(H^{(h)})\le n$.

If every entity occurs as the head of at least one triple, then for
each entity $e_i$ there exists a column whose non-zero entry occurs in
row $i$. The rows of $H^{(h)}$ therefore span $\mathbb{R}^n$, giving
$\operatorname{rank}(H^{(h)})=n$.

The same argument applies to $H^{(t)}$.
\end{proof}

\section{Knowledge Graphs as Free Categories}\label{sec:freecat}

\begin{definition}[Free category of a knowledge graph]\label{def:freecat}
The \emph{free category} $\mathcal{C}(K)$ generated by $K$ has
\begin{itemize}
\item objects: the entities $E$;
\item generating morphisms: the triples $T$, with $\tau = (h,p,t)$
      regarded as an arrow $h \to t$;
\item morphisms: finite composable paths of triples, with composition
      given by concatenation;
\item identities: empty paths at each entity.
\end{itemize}
\end{definition}

This construction is the standard free category generated by a directed
graph (cf.~\cite[Sec.~II.7]{maclane1998cwm}; \cite[Ch.~4]{awodey2010}).

\begin{theorem}[Universal property of the free category]
\label{thm:freecat-universal}
Let $K=(E,P,T)$ be a knowledge graph and let $\mathcal{C}(K)$ be the
free category generated by $K$. Let $\mathcal{D}$ be any category.

Suppose we are given a function assigning to each entity
$e\in E$ an object $F(e)$ of $\mathcal{D}$ and to each triple
$\tau=(h,p,t)\in T$ a morphism
\[
F(\tau):F(h)\to F(t)
\]
in $\mathcal{D}$.

Then there exists a unique functor
\[
\widetilde{F}:\mathcal{C}(K)\to\mathcal{D}
\]
such that

\[
\widetilde{F}(e)=F(e),\qquad
\widetilde{F}(\tau)=F(\tau)
\]

\noindent for all entities $e\in E$ and triples $\tau\in T$.
\end{theorem}

\begin{proof}
Every morphism in $\mathcal{C}(K)$ is a finite composable path of
triples
\[
\tau_k\circ\dots\circ\tau_1 .
\]

Define
\[
\widetilde{F}(\tau_k\circ\dots\circ\tau_1)
=
F(\tau_k)\circ\dots\circ F(\tau_1).
\]

This definition respects identities and composition because
composition in $\mathcal{C}(K)$ is given by concatenation of paths.
Hence $\widetilde{F}$ is a functor extending the assignment on
generators.

Uniqueness follows because every morphism of $\mathcal{C}(K)$ is
generated by the triples of $T$, so any functor agreeing with $F$ on
entities and triples must coincide with $\widetilde{F}$.
\end{proof}

\begin{proposition}\label{prop:domcod}
Let $K=(E,P,T)$ be a knowledge graph and let $\mathcal{C}(K)$ be its
free category. For every triple $\tau=(h,p,t)\in T$, the corresponding
generating morphism in $\mathcal{C}(K)$ has
\[
\mathrm{dom}(\tau)=h,
\qquad
\mathrm{cod}(\tau)=t .
\]
Moreover, the maps
\[
h,t : T \to E
\]
coincide with the domain and codomain maps of the generating morphisms
of $\mathcal{C}(K)$.
\end{proposition}

\begin{proof}
By Definition~\ref{def:freecat}, each triple $\tau=(h,p,t)$ is regarded
as an arrow $h\to t$ in $\mathcal{C}(K)$. Therefore the domain and
codomain of the corresponding generating morphism are exactly $h$ and
$t$. Since the functions $h,t:T\to E$ assign to each triple its head
and tail entities, they coincide with the domain and codomain maps of
the generating morphisms.
\end{proof}

The out-line and in-line digraphs record co-domain and co-codomain
adjacencies among the generating morphisms: $\tau_i$ and $\tau_j$ are
adjacent in $L_{\mathrm{out}}(K)$ iff they share the same domain, and
in $L_{\mathrm{in}}(K)$ iff they share the same codomain.

\begin{proposition}[Domain and codomain fibres]
\label{prop:fibre-structure}

Let $K=(E,P,T)$ be a knowledge graph and let $\mathcal{C}(K)$ be the
free category generated by $K$.

For each entity $e\in E$ define the sets

\[
\mathrm{Hom}_h(e)=\{\tau\in T : h(\tau)=e\},
\qquad
\mathrm{Hom}_t(e)=\{\tau\in T : t(\tau)=e\}.
\]

Then the generating morphisms of $\mathcal{C}(K)$ decompose into
domain and codomain fibres

\[
T=\bigsqcup_{e\in E}\mathrm{Hom}_h(e)
=
\bigsqcup_{e\in E}\mathrm{Hom}_t(e).
\]

Moreover, the connected components of the out-line and in-line
knowledge digraphs correspond exactly to these domain and codomain
fibres.
\end{proposition}

\begin{proof}
By definition each triple $\tau=(h,p,t)$ has a unique head and tail.
Therefore the sets $\mathrm{Hom}_h(e)$ and $\mathrm{Hom}_t(e)$ form
partitions of the triple set $T$.

From the construction of the line knowledge digraphs,
two triples are adjacent in $L_{\mathrm{out}}(K)$ precisely when
they share the same head entity, and adjacent in $L_{\mathrm{in}}(K)$
when they share the same tail entity. Hence the connected components
of the respective line digraphs coincide with the fibres
$\mathrm{Hom}_h(e)$ and $\mathrm{Hom}_t(e)$.
\end{proof}

The categorical representation of knowledge graphs makes it possible to
study transformations between graphs using categorical morphisms.
Mappings between knowledge graphs that preserve entities, predicates,
and triples naturally induce structure-preserving maps between their
associated free categories. Such constructions align with the broader
categorical treatment of relational data and database transformations
developed in functorial data migration \cite{spivak2012,spivak2014category}.

Functorial data migration interprets database schemas as categories
and database instances as functors into the category of sets.
Schema mappings are then represented by functors between categories,
which induce adjoint data-migration operations between instance
categories \cite{spivak2012}. In the present setting, a knowledge graph
generates the free category $\mathcal{C}(K)$, which may be regarded as
a categorical schema whose morphisms correspond to relational paths.
Sheaves on $(\mathcal{C}(K),J)$ can therefore be interpreted as
structured instances of relational data satisfying compatibility
conditions imposed by the topology $J$.

From this viewpoint the geometric morphisms between the associated
topoi describe transformations between semantic environments attached
to the same relational schema. The topos
$\mathrm{Sh}(\mathcal{C}(K),J)$ thus provides a semantic extension of
the categorical data model in which contextual interpretation is
encoded through the Grothendieck topology.

\begin{definition}[Knowledge graph homomorphism]
\label{def:kg-hom}
Let $K=(E,P,T)$ and $K'=(E',P',T')$ be knowledge graphs.
A \emph{knowledge graph homomorphism}
\[
f:K\to K'
\]
consists of maps

\[
f_E:E\to E', \qquad
f_P:P\to P'
\]

\noindent such that for every triple
\[
(h,p,t)\in T
\]
the triple

\[
(f_E(h),\,f_P(p),\,f_E(t))
\]

\noindent belongs to $T'$.

In other words, triples are mapped to triples while preserving the
head entity, predicate label, and tail entity.
\end{definition}

\begin{proposition}
\label{prop:kg-category}
Knowledge graphs together with knowledge graph homomorphisms form a
category, denoted $\mathbf{KG}$.
\end{proposition}

\begin{proof}
The identity morphism on a knowledge graph $K=(E,P,T)$ is given by the
identity maps on $E$ and $P$. Composition of two homomorphisms
\[
K\xrightarrow{f}K'\xrightarrow{g}K''
\]
is defined componentwise by composition of the entity and predicate
maps. Preservation of triples follows immediately from the definition.
Associativity and identity laws hold because they hold for functions.
\end{proof}

\subsection{Functorial constructions}\label{subsec:functors}

Let $\mathbf{KG}$ denote the category of knowledge graphs, whose
morphisms $f : K \to K'$ consist of maps $f_E : E \to E'$ and
$f_T : T \to T'$ preserving heads, tails, and predicates.
Such categorical interpretations of relational data are closely related
to functorial approaches to data migration
\cite{spivak2012,spivak2014category}. Let $\mathbf{Graph}$ be the category of directed graphs and graph
homomorphisms.

\begin{theorem}[Line digraph functors]\label{thm:functor}
The assignments $L_{\mathrm{out}}, L_{\mathrm{in}} : \mathbf{KG} \to
\mathbf{Graph}$ extend to functors.
\end{theorem}

This universal property characterises the free category generated by a
directed graph (see \cite[Sec.~II.7]{maclane1998cwm}; \cite[Ch.~4]{awodey2010}).

\begin{proof}
We treat $L_{\mathrm{out}}$; the argument for $L_{\mathrm{in}}$ is
analogous. Given a morphism $f : K \to K'$ in $\mathbf{KG}$, define
\[
  L_{\mathrm{out}}(f) : T \to T',
  \qquad
  L_{\mathrm{out}}(f)(\tau_j) := f_T(\tau_j).
\]
To verify this is a graph homomorphism, suppose $\tau_i \to \tau_j$ is
an edge in $L_{\mathrm{out}}(K)$, i.e.\ $h(\tau_i) = h(\tau_j)$.
Since $f$ preserves heads,
\[
  h(f_T(\tau_i)) = f_E(h(\tau_i)) = f_E(h(\tau_j)) = h(f_T(\tau_j)),
\]
so $f_T(\tau_i) \to f_T(\tau_j)$ is an edge in $L_{\mathrm{out}}(K')$.
Preservation of identities and composition follows from the fact that
$f_T$ is a function. Hence $L_{\mathrm{out}}$ is a functor.
\end{proof}

\section{Sites on Knowledge Graph Categories}\label{sec:sites}

To introduce semantic structure we equip $\mathcal{C}(K)$ with a
Grothendieck topology~\cite{maclane1992sheaves,johnstone2002elephant}.
Grothendieck topologies provide a general framework for defining
notions of covering and locality in categorical settings. In the
present context, the topology will encode how relational information
propagates across the knowledge graph through composable paths. The
resulting site structure allows semantic information to be assigned
locally and then extended globally through the sheaf condition
\cite{maclane1992sheaves,johnstone2002elephant}.

\begin{definition}[Path-covering topology]\label{def:topology}
For an entity $e \in E$ (viewed as an object of $\mathcal{C}(K)$), a
\emph{covering family} is a finite set of morphisms
$\{f_i : e_i \to e\}_{i \in I}$ in $\mathcal{C}(K)$ such that every
entity $e' \in E$ reachable from $e$ by a relational path is also
reachable from some $e_i$ via a path that factors through $e$.
\end{definition}

We denote by $J$ the Grothendieck topology on $\mathcal{C}(K)$ determined
by the covering families of Definition~\ref{def:topology}.

\begin{remark}[Interpretation of the topology]
The path-covering topology encodes how relational information may
propagate through the knowledge graph via composable triples. In this
sense it specifies which families of morphisms provide sufficient
local information to determine semantic interpretations at a given
entity. This interpretation parallels the role of open coverings in
topology, where local data defined on open sets determine global
structures when suitable compatibility conditions are satisfied
\cite{maclane1992sheaves,johnstone2002elephant}.
\end{remark}

\begin{proposition}[Grothendieck topology axioms]\label{prop:topology}
The covering families of Definition~\ref{def:topology} satisfy the three
axioms of a Grothendieck topology: maximality, stability under pullback,
and transitivity.
\end{proposition}

\begin{proof}
\emph{Maximality.} The identity morphism $\mathrm{id}_e : e \to e$
forms a singleton covering family, since every path through $e$ factors
through the identity.

\emph{Stability under pullback.}
Let $\{f_i : e_i \to e\}$ be a covering family and let
$g : e' \to e$ be any morphism in $\mathcal{C}(K)$.
In the free category $\mathcal{C}(K)$, morphisms correspond to
finite relational paths in the knowledge graph.

Consider an entity $e''$ reachable from $e'$ by a path
$e'' \to e'$. Composing with $g$ gives a path
\[
e'' \to e' \xrightarrow{g} e .
\]
Since $\{f_i : e_i \to e\}$ is a covering family, this composite
path factors through some morphism $f_i : e_i \to e$.
Hence the original path from $e''$ factors through the
pullback of $f_i$ along $g$.

Therefore every entity reachable from $e'$ factors through one
of the pulled-back morphisms, and the resulting family forms
a covering family of $e'$.

Formally, the path-covering condition defines a Grothendieck coverage
whose generated topology $J$ satisfies the axioms of a Grothendieck
topology without requiring the existence of all pullbacks
\cite{maclane1992sheaves,johnstone2002elephant}.

\emph{Transitivity.} Suppose $\{f_i : e_i \to e\}$ is a covering family
and that for each $i$ we have a covering family
$\{g_{ij} : e_{ij} \to e_i\}$. Consider the composites
\[
f_i \circ g_{ij} : e_{ij} \to e .
\]
Any entity reachable from $e$ is reachable from some $e_i$, and from
there from some $e_{ij}$ via the second covering family. Hence the
family $\{f_i \circ g_{ij} : e_{ij} \to e\}$ satisfies the covering
condition.
\end{proof}

We denote the resulting site by $(\mathcal{C}(K), J)$.

\subsection{Alternative interpretive topologies}

The path-covering topology introduced above captures relational
propagation along composable paths of triples. It is useful to
contrast this topology with a more restrictive interpretation in which
no contextual propagation occurs.

\begin{definition}[Atomic topology]\label{def:atomic-topology}
Let $\mathcal{C}(K)$ be the free category generated by a knowledge
graph $K$. The \emph{atomic topology} $J_{\mathrm{atom}}$ on
$\mathcal{C}(K)$ is the Grothendieck topology in which a covering
family of an object $e$ consists only of isomorphisms
\[
f:e'\to e .
\]
\end{definition}

In this topology, covering families express only identity-level
equivalences between objects and therefore do not propagate relational
information along paths of triples.
The topology $J_{\mathrm{atom}}$ represents a purely local interpretation of triples,
while the path-covering topology $J$ encodes contextual propagation of
relational information along composable paths.

The two topologies lead to two distinct site structures on the same
underlying category:

\[
S_{\mathrm{path}}=(\mathcal{C}(K),J),
\qquad
S_{\mathrm{atom}}=(\mathcal{C}(K),J_{\mathrm{atom}}).
\]

The topology $J$ reflects a contextual interpretation of the knowledge
graph in which relational paths propagate semantic information, while
the topology $J_{\mathrm{atom}}$ represents a purely local
interpretation in which objects are interpreted independently of such
propagation.

These two sites therefore correspond to two different logical
environments in which the same combinatorial knowledge graph can be
interpreted.

\begin{proposition}\label{prop:site-functoriality}
Let $f:K\to K'$ be a knowledge graph homomorphism in the sense of
Definition~\ref{def:kg-hom}. Then the induced functor

\[
\mathcal{C}(f):\mathcal{C}(K)\to\mathcal{C}(K')
\]

\noindent between the free categories preserves covering families of the
path-covering topology. Consequently $\mathcal{C}(f)$ defines a
morphism of sites

\[
(\mathcal{C}(K),J)\longrightarrow(\mathcal{C}(K'),J').
\]
\end{proposition}

\begin{proof}
A covering family of an object $e$ in $\mathcal{C}(K)$ consists of
morphisms $\{f_i:e_i\to e\}$ such that every entity reachable from $e$
through a relational path factors through one of the $f_i$.
Since a knowledge graph homomorphism preserves heads, tails, and
predicates of triples, the induced functor $\mathcal{C}(f)$ maps
relational paths in $K$ to relational paths in $K'$.
Therefore reachability relations between entities are preserved.
Hence the image of a covering family under $\mathcal{C}(f)$ again
satisfies the covering condition in $(\mathcal{C}(K'),J')$.
\end{proof}

\begin{proposition}\label{prop:topology-inclusion}
Let $J_{\mathrm{atom}}$ be the atomic topology and $J$ the
path-covering topology on $\mathcal{C}(K)$.

Then every covering family in $J_{\mathrm{atom}}$ is also a covering
family in $J$. Consequently, the identity functor

\[
\mathrm{id}_{\mathcal{C}(K)} : (\mathcal{C}(K),J)
\longrightarrow
(\mathcal{C}(K),J_{\mathrm{atom}})
\]

\noindent defines a morphism of sites.
\end{proposition}

\begin{proof}
In the atomic topology $J_{\mathrm{atom}}$, the only covering families
consist of isomorphisms $f:e'\to e$. Such morphisms trivially satisfy
the path-covering condition of Definition~\ref{def:topology} because
any path beginning at $e$ factors through the identity morphism of $e$.
Hence every covering family of $J_{\mathrm{atom}}$ is also a covering
family of $J$.

Therefore the identity functor on $\mathcal{C}(K)$ sends coverings of
$J_{\mathrm{atom}}$ to coverings of $J$, and thus defines a morphism
of sites \cite{maclane1992sheaves,johnstone2002elephant}.
\end{proof}

\section{Sheaves and Topos Construction}\label{sec:topos}

\begin{definition}[Presheaf and sheaf]\label{def:sheaf}
A \emph{presheaf} on $(\mathcal{C}(K), J)$ is a contravariant functor
$F : \mathcal{C}(K)^{\mathrm{op}} \to \mathbf{Set}$. It is a
\emph{sheaf} if for every object $e$ and every covering family
$\{f_i : e_i \to e\}$, the diagram
\[
  F(e) \;\longrightarrow\; \prod_i F(e_i)
  \;\rightrightarrows\; \prod_{i,j} F(e_i \times_e e_j)
\]
is an equaliser.
\end{definition}

Presheaf categories of the form
\[
[\mathcal{C}^{\mathrm{op}},\mathbf{Set}]
\]
provide canonical examples of topoi in category theory
\cite{maclane1992sheaves,johnstone2002elephant}. The sheaf condition
strengthens the presheaf structure by requiring that compatible local
data defined on covering families glue uniquely to global data. In the
context of knowledge graphs this property expresses the principle that
locally consistent semantic information can be assembled into a global
interpretation of the relational structure.

\begin{theorem}[Sheaf topos]\label{thm:topos}
The category $\mathrm{Sh}(\mathcal{C}(K), J)$ of sheaves on the site
$(\mathcal{C}(K), J)$ is a Grothendieck topos.
\end{theorem}

\begin{proof}
By Giraud's theorem~\cite{johnstone2002elephant}, a category is a
Grothendieck topos if and only if it is a left-exact localisation of a
presheaf category. The category of presheaves
\[
\widehat{\mathcal{C}(K)} =
[\mathcal{C}(K)^{\mathrm{op}},\mathbf{Set}]
\]
is a topos. The sheafification functor
\[
a:\widehat{\mathcal{C}(K)} \to \mathrm{Sh}(\mathcal{C}(K),J)
\]
is left adjoint to the inclusion
\[
i:\mathrm{Sh}(\mathcal{C}(K),J) \hookrightarrow
\widehat{\mathcal{C}(K)}
\]
and the functor $a$ is left exact. Hence
$\mathrm{Sh}(\mathcal{C}(K),J)$ is a Grothendieck topos.
\end{proof}

\begin{proposition}\label{prop:two-topoi}
Let $J$ be the path-covering topology of
Definition~\ref{def:topology} and let $J_{\mathrm{atom}}$
be the atomic topology of
Definition~\ref{def:atomic-topology}.

Then the sites

\[
(\mathcal{C}(K),J)
\qquad\text{and}\qquad
(\mathcal{C}(K),J_{\mathrm{atom}})
\]

\noindent give rise to two Grothendieck topoi

\[
\mathrm{Sh}(\mathcal{C}(K),J)
\qquad\text{and}\qquad
\mathrm{Sh}(\mathcal{C}(K),J_{\mathrm{atom}}).
\]
\end{proposition}

\begin{proof}
For any site $(\mathcal{C},J)$ the category of sheaves
$\mathrm{Sh}(\mathcal{C},J)$ is a Grothendieck topos
\cite{maclane1992sheaves,johnstone2002elephant}.
Applying this general result to the two Grothendieck
topologies $J$ and $J_{\mathrm{atom}}$ on the same category
$\mathcal{C}(K)$ yields the two stated topoi.
\end{proof}

\begin{theorem}\label{thm:geom-morphism}
Let $f:K\to K'$ be a knowledge graph homomorphism and let

\[
\mathcal{C}(f):\mathcal{C}(K)\to\mathcal{C}(K')
\]

\noindent be the induced functor between the free categories.  
Then $f$ induces a geometric morphism between the associated
sheaf topoi

\[
\mathrm{Sh}(\mathcal{C}(K),J)
\longrightarrow
\mathrm{Sh}(\mathcal{C}(K'),J').
\]
\end{theorem}

\begin{proof}
By Proposition~\ref{prop:site-functoriality}, the functor
$\mathcal{C}(f)$ defines a morphism of sites

\[
(\mathcal{C}(K),J)\to(\mathcal{C}(K'),J').
\]

A morphism of sites induces an adjoint pair of functors between the
associated sheaf categories

\[
f^*:\mathrm{Sh}(\mathcal{C}(K'),J')\to
\mathrm{Sh}(\mathcal{C}(K),J),
\qquad
f_*:\mathrm{Sh}(\mathcal{C}(K),J)\to
\mathrm{Sh}(\mathcal{C}(K'),J').
\]

The functor $f^*$ is obtained by precomposition with
$\mathcal{C}(f)$ followed by sheafification and is left exact.
Hence $(f^*,f_*)$ forms a geometric morphism of topoi.
\end{proof}

\begin{theorem}[Geometric morphism between interpretive topoi]
\label{thm:interpretive-morphism}

Let $J$ be the path-covering topology and
$J_{\mathrm{atom}}$ the atomic topology on $\mathcal{C}(K)$.

Then the identity functor on $\mathcal{C}(K)$ induces a geometric
morphism

\[
\mathrm{Sh}(\mathcal{C}(K),J)
\longrightarrow
\mathrm{Sh}(\mathcal{C}(K),J_{\mathrm{atom}}).
\]

\end{theorem}

\begin{proof}
Let
\[
f : (\mathcal{C}(K),J) \longrightarrow (\mathcal{C}(K),J_{\mathrm{atom}})
\]
be the morphism of sites induced by the identity functor on
$\mathcal{C}(K)$.

Since every $J_{\mathrm{atom}}$-covering sieve is trivial, the
identity functor is continuous with respect to the two topologies.
Hence it defines a morphism of sites.

The associated inverse-image functor
\[
g^* :
\mathrm{Sh}(\mathcal{C}(K),J_{\mathrm{atom}})
\to
\mathrm{Sh}(\mathcal{C}(K),J)
\]
is given by precomposition with the identity functor followed by
sheafification with respect to $J$.

The direct-image functor
\[
g_* :
\mathrm{Sh}(\mathcal{C}(K),J)
\to
\mathrm{Sh}(\mathcal{C}(K),J_{\mathrm{atom}})
\]
is given by restriction along the identity functor.

Since precomposition preserves finite limits, the inverse-image
functor $g^*$ is left exact. Therefore $(g^*,g_*)$ defines a
geometric morphism between the two sheaf topoi
(see \cite[Ch.~III]{maclane1992sheaves};
\cite[Ch.~C2]{johnstone2002elephant}).
\end{proof}

\begin{proposition}[Essential geometric morphism]
\label{prop:essential}

The geometric morphism of
Theorem~\ref{thm:interpretive-morphism}

\[
\mathrm{Sh}(\mathcal{C}(K),J)
\longrightarrow
\mathrm{Sh}(\mathcal{C}(K),J_{\mathrm{atom}})
\]

\noindent is essential.

\end{proposition}

\begin{proof}
Let
\[
g :
\mathrm{Sh}(\mathcal{C}(K),J)
\longrightarrow
\mathrm{Sh}(\mathcal{C}(K),J_{\mathrm{atom}})
\]
be the geometric morphism of
Theorem~\ref{thm:interpretive-morphism}.

Its inverse–image functor
\[
g^* :
\mathrm{Sh}(\mathcal{C}(K),J_{\mathrm{atom}})
\to
\mathrm{Sh}(\mathcal{C}(K),J)
\]
is induced by precomposition with the identity functor on
$\mathcal{C}(K)$ followed by sheafification.

Since precomposition with a functor between small categories
admits left Kan extensions, the functor $g^*$ has a left adjoint

\[
g_! :
\mathrm{Sh}(\mathcal{C}(K),J)
\to
\mathrm{Sh}(\mathcal{C}(K),J_{\mathrm{atom}}).
\]

Therefore the geometric morphism is essential
(see \cite[Ch.~III]{maclane1992sheaves};
\cite[Ch.~C2]{johnstone2002elephant}).
\end{proof}











\begin{remark}[Interpretation of the adjunction]
\label{rem:adjunction-interpretation}

The essential geometric morphism constructed above determines an
adjoint triple of functors

\[
g_! \;\dashv\; g^* \;\dashv\; g_* .
\]

These functors describe three complementary semantic operations
between the two interpretive environments.

The inverse-image functor $g^*$ transports interpretations from the
atomic topos to the path-topos, allowing locally defined meanings to
propagate along relational paths. The direct-image functor $g_*$
aggregates contextual interpretations back into the purely local
setting of the atomic topology. The additional left adjoint $g_!$
corresponds to the free extension of locally specified information
into the richer contextual environment.

In this sense, the adjoint triple formalises three modes of semantic
translation between the two regimes of interpretation encoded by the
two Grothendieck topologies.
\end{remark}

Topos theory therefore equips the categorical representation of a
knowledge graph with an internal logical structure. In particular,
every Grothendieck topos supports an internal higher-order
intuitionistic logic in which propositions correspond to subobjects
and logical operations arise from categorical constructions
\cite{maclane1992sheaves,johnstone1977,hogan2021kg,johnstone2002elephant}. This logical
interpretation provides a natural setting for formal reasoning about
context-dependent relational information encoded in knowledge graphs.

\subsection{Subobject classifier and internal logic}\label{subsec:logic}

Topos theory provides a setting in which logical reasoning can be
carried out inside a categorical structure \cite{freyd1990}. In particular, every Grothendieck topos admits an internal higher-order intuitionistic
logic whose truth values are represented by the subobject classifier
$\Omega$. This internal logic allows propositions about objects of the
topos to be interpreted categorically through morphisms and subobjects
\cite{maclane1992sheaves,johnstone2002elephant}.

Every topos has a \emph{subobject classifier} $\Omega$, a distinguished
object such that subobjects of any object $X$ correspond naturally to
morphisms $X \to \Omega$~\cite{maclane1992sheaves}. In
$\mathrm{Sh}(\mathcal{C}(K), J)$, the subobject classifier assigns to
each entity $e$ the set of \emph{sieves} on $e$ that belong to the
topology $J$. The internal logic of the topos is in general
intuitionistic: the law of excluded middle may fail, reflecting the
fact that truth values are context-dependent and not globally
determined.

\section{Semantic Interpretation}\label{sec:semantics}

The topos $\mathrm{Sh}(\mathcal{C}(K), J)$ provides a formal setting
for context-dependent semantics of knowledge graphs. 
Similar sheaf-theoretic approaches have been used in several areas of
mathematics and theoretical computer science to model distributed or
contextual information. In such settings, local data assigned to
objects of a category are required to satisfy compatibility conditions
on overlaps so that they can be assembled into coherent global
structures \cite{maclane1992sheaves,johnstone2002elephant}. The same
principle applies here: semantic information attached to entities and
relations in a knowledge graph can be organised locally and then
integrated through the sheaf structure.
Objects of this
topos may be interpreted as semantic structures defined over the
categorical representation of the knowledge graph. In this setting,
semantic information is assigned locally and propagated along
relational paths through the functorial structure of sheaves.

A useful conceptual distinction between the combinatorial structure of
the graph and its semantic interpretation can be articulated through
the philosophical distinction between \emph{being} and \emph{appearing}
developed in philosophical discussions of ontology and the logics of appearance (\cite{badiou2006}). In this terminology the
underlying knowledge graph records the \emph{being} of relational data:
entities and triples exist as combinatorial objects. The associated
sheaf topos, by contrast, provides a mathematical environment in which
these relations \emph{appear} through context-dependent interpretations.
Different Grothendieck topologies therefore correspond to different
regimes of appearance, determining how relational information may be
locally interpreted and globally integrated.
In Badiou's sense, the topos does not merely represent the
data of the graph but organises the ways in which that data can appear
through compatible local interpretations.

The semantic role of sheaves can be understood as assigning
context-dependent information to the entities of the knowledge graph.
Each entity carries a set of possible interpretations, while relational
paths induce restriction maps between these sets. In this way the
sheaf structure formalises how semantic information propagates through
the relational structure of the graph and how locally defined meanings
must agree when combined into global interpretations
\cite{maclane1992sheaves,johnstone2002elephant}.

\begin{proposition}\label{prop:sections}
Let $F$ be a sheaf on the site $(\mathcal{C}(K),J)$. For every morphism
$f:e' \to e$ in $\mathcal{C}(K)$ there is a restriction map
\[
F(f):F(e)\to F(e')
\]
such that the assignment $e\mapsto F(e)$ and $f\mapsto F(f)$ defines a
contravariant functor from $\mathcal{C}(K)$ to $\mathbf{Set}$.
\end{proposition}

\begin{proof}
By Definition~\ref{def:sheaf}, a sheaf is a contravariant functor
$F:\mathcal{C}(K)^{\mathrm{op}}\to\mathbf{Set}$ satisfying the sheaf
condition with respect to the topology $J$. Therefore each morphism
$f:e'\to e$ induces a map $F(f):F(e)\to F(e')$. Functoriality implies
that $F(\mathrm{id}_e)=\mathrm{id}_{F(e)}$ and
\[
F(g\circ f)=F(f)\circ F(g)
\]
for any composable morphisms $e''\xrightarrow{g}e'\xrightarrow{f}e$.
Hence the assignment is a contravariant functor.
\end{proof}

\subsection{Sections as interpretations}

A sheaf $F$ assigns to each entity $e$ a set $F(e)$ of
\emph{local sections}, representing the interpretations or meanings
available at $e$. The restriction maps $F(f):F(e)\to F(e')$ for a
morphism $f:e'\to e$ encode how semantic information propagates along
relational paths in the knowledge graph.

\subsection{Contextual truth}

The interpretation of truth in a topos differs from classical logic.
Instead of a two-valued truth set, propositions take values in the
subobject classifier $\Omega$, which encodes degrees of validity
relative to the covering structure of the site. Consequently,
statements about objects of the knowledge graph may hold locally in
certain contexts while failing globally, reflecting the inherently
context-dependent nature of relational information
\cite{maclane1992sheaves,johnstone2002elephant}.

A \emph{global section} of $F$ is an element of $F(e_{\mathrm{top}})$
for a terminal object when such an object exists, or equivalently a
compatible family of local sections defined across all entities of the
site. The sheaf condition ensures that locally consistent
interpretations can be uniquely assembled into a global one. In this
sense, a semantic statement about the knowledge graph is globally true
only when it is locally valid in each relevant context and these local
interpretations agree on overlaps.

\subsection{Local-to-global reasoning}

The gluing axiom of Definition~\ref{def:sheaf} expresses the
local-to-global principle: compatible local data can be uniquely
assembled into global data. In the context of knowledge graphs this
means that if a semantic assignment—such as an iconographic meaning, a
stylistic classification, or a confidence score—is defined consistently
on each part of a covering of an entity, then it extends uniquely to
the whole entity.

For instance, on a simple graph $A \to B \leftarrow C$, a sheaf of
interpretations assigns sets $F(A),F(B),F(C)$ with restriction maps
along the two arrows, allowing compatible local interpretations at
$A$ and $C$ to glue into a global interpretation at $B$.

\subsection{Change of semantic context}

The existence of the geometric morphism of
Theorem~\ref{thm:interpretive-morphism} shows that the same knowledge
graph may support different semantic interpretations depending on the
Grothendieck topology imposed on its free category.
The geometric morphism therefore formalizes a change of semantic regime,
transporting interpretations between the purely local environment and the
contextual environment generated by relational propagation.

The atomic topology corresponds to a strictly local interpretation in
which entities are interpreted independently of relational propagation.
By contrast, the path-covering topology allows semantic information to
propagate along composable relational paths. The geometric morphism
between the associated topoi therefore formalises the passage from a
local interpretation of relational data to a context-sensitive one in
which meanings can be integrated across relational neighbourhoods.

From this perspective, the topos
\[
\mathrm{Sh}(\mathcal{C}(K),J)
\]
encodes a richer semantic environment than
\[
\mathrm{Sh}(\mathcal{C}(K),J_{\mathrm{atom}}),
\]
since the topology $J$ introduces additional covering relations that
allow information to be aggregated along paths of triples.

\subsection{Change of topos and semantic events}

The passage between the two semantic environments introduced in
Section~\ref{sec:sites} can also be interpreted as a transition
between regimes of appearance. In the language of category theory,
this transition is formalised by the geometric morphism of
Theorem~\ref{thm:interpretive-morphism} relating the topoi
\[
\mathrm{Sh}(\mathcal{C}(K),J)
\quad\text{and}\quad
\mathrm{Sh}(\mathcal{C}(K),J_{\mathrm{atom}}).
\]

Conceptually, such a transition may be understood as a change in the
logical environment in which relational information is interpreted.
Where the atomic topology encodes a strictly local interpretation of
entities and relations, the path-covering topology introduces a richer
context in which semantic information propagates along relational
paths.

This perspective resonates with the interpretation of topoi as
structures governing regimes of appearance in mathematical ontology
\cite{badiou2006}. In this sense, modifying the Grothendieck topology
changes the conditions under which relational statements become
semantically meaningful, thereby producing a new interpretive
environment for the same underlying relational data.

\section{Illustrative Example}\label{sec:example}

We illustrate the framework developed in the previous sections on a
small knowledge graph and show explicitly how the constructions of
incidence matrices, line digraphs, free categories and sheaves appear
in a concrete case.

\subsection{Knowledge graph}

Let
\[
E=\{A,B,C,D\}
\]
and
\[
T=\{\tau_1=(A,r_1,B),\;
     \tau_2=(A,r_2,C),\;
     \tau_3=(D,r_3,B),\;
     \tau_4=(D,r_4,C)\}.
\]

Thus the graph contains two triples with head $A$ and two triples with
head $D$.

\subsection{Incidence matrices}

With the orderings $(A,B,C,D)$ and
$(\tau_1,\tau_2,\tau_3,\tau_4)$ the head and tail incidence matrices are

\[
H^{(h)}=
\begin{pmatrix}
1 & 1 & 0 & 0\\
0 & 0 & 0 & 0\\
0 & 0 & 0 & 0\\
0 & 0 & 1 & 1
\end{pmatrix},
\qquad
H^{(t)}=
\begin{pmatrix}
0 & 0 & 0 & 0\\
1 & 0 & 1 & 0\\
0 & 1 & 0 & 1\\
0 & 0 & 0 & 0
\end{pmatrix}.
\]

\subsection{Line digraphs}

Computing
\[
M_{\mathrm{out}}=(H^{(h)})^\top H^{(h)}
\]
and setting the diagonal to zero gives

\[
A_{\mathrm{out}}=
\begin{pmatrix}
0 & 1 & 0 & 0\\
1 & 0 & 0 & 0\\
0 & 0 & 0 & 1\\
0 & 0 & 1 & 0
\end{pmatrix}.
\]

This example illustrates concretely the structural decomposition
established in Section~\ref{sec:line}. Triples sharing the same head
entity generate strongly connected components in the out-line digraph.
In the present case the triples $\tau_1$ and $\tau_2$ share the head
$A$, while $\tau_3$ and $\tau_4$ share the head $D$. Consequently the
graph $L_{\mathrm{out}}(K)$ splits into two components determined by
these shared-head relations.

Such decompositions are closely related to classical line graph
constructions studied in graph theory, where vertices correspond to
edges of the original graph and adjacency reflects incidence relations
between edges \cite{beineke1970,harary1969}. The knowledge-graph
formulation adapts this idea by treating triples as vertices and
identifying adjacency through shared head or tail entities.

\begin{proposition}\label{prop:example-components}
The graph $L_{\mathrm{out}}(K)$ consists of two strongly connected
components,
\[
\{\tau_1,\tau_2\}
\quad\text{and}\quad
\{\tau_3,\tau_4\},
\]
each forming a complete directed graph.
\end{proposition}

\begin{proof}
This follows directly from the structural decomposition established in
Section~\ref{sec:line}. In the present graph the triples
$\tau_1,\tau_2$ share the head $A$ and the triples
$\tau_3,\tau_4$ share the head $D$, so the strongly connected
components are exactly the two corresponding equivalence classes.
\end{proof}

\subsection{Free category}

The free category $\mathcal{C}(K)$ has objects
\[
\{A,B,C,D\}
\]
and four generating morphisms
\[
\tau_1:A\to B,\quad
\tau_2:A\to C,\quad
\tau_3:D\to B,\quad
\tau_4:D\to C .
\]
Since no triple ends at the starting point of another triple,
no nontrivial compositions occur, and the only non-identity morphisms
are the generating ones.

The categorical structure of this example can also be represented
diagrammatically. The generating morphisms correspond to arrows between
objects of the free category:

\[
\begin{tikzcd}
A \arrow[r,"\tau_1"] & B \\
A \arrow[r,"\tau_2"] & C \\
D \arrow[r,"\tau_3"] & B \\
D \arrow[r,"\tau_4"] & C
\end{tikzcd}
\]

Such diagrammatic representations are standard in category theory for
visualising morphisms and their relations
\cite{maclane1998cwm,awodey2010,riehl2016}.

\subsection{Covering and sheaf}

Consider the object $B$ in $\mathcal{C}(K)$.
The morphisms
\[
\tau_1:A\to B,
\qquad
\tau_3:D\to B
\]
form a covering family of $B$.

Let $F$ be a sheaf on $(\mathcal{C}(K),J)$.
The sheaf assigns sets
\[
F(A),\quad F(D),\quad F(B)
\]
together with restriction maps induced by the morphisms
$\tau_1$ and $\tau_3$.

The sheaf condition requires that compatible local sections
\[
s_1\in F(A),\qquad s_3\in F(D)
\]
glue to a unique element of $F(B)$.
In this example the pullback
\[
A\times_B D
\]
is empty, so any pair $(s_1,s_3)$ is compatible.
Consequently the set of sections at $B$ is naturally identified with
the product
\[
F(B)\cong F(A)\times F(D).
\]

Thus the interpretation associated with the entity $B$ is determined
by the local interpretations attached to the entities $A$ and $D$.

\section{Conclusion}\label{sec:conclusion}

The framework therefore connects three levels of structure in knowledge graphs: combinatorial
graph structure, categorical composition in the free category, and
topos-theoretic semantics supporting local-to-global reasoning.
The incidence representation provides an algebraic description of the
interaction between entities and triples, the associated line knowledge
digraphs reveal structural relations between triples themselves, and
the categorical construction of the free category $\mathcal{C}(K)$
allows relational paths to be interpreted through categorical
composition. These constructions together provide a bridge between
classical graph-theoretic methods and categorical semantics
\cite{harary1969,maclane1998cwm}.
The geometric morphism, in particular, formalizes a change of semantic regime,
transporting interpretations between the purely local environment and the
contextual environment generated by relational propagation.

As a matter of fact, we have developed a categorical framework for knowledge graphs that
integrates three complementary levels of structure. At the combinatorial
level, knowledge graphs were represented through head and tail incidence
matrices and analysed using the associated line knowledge digraph
constructions. At the categorical level, the relational structure of a
knowledge graph was interpreted through the free category generated by
its entities and triples, allowing compositional reasoning about
relational paths. At the semantic level, the categorical structure was
equipped with a Grothendieck topology whose associated category of
sheaves forms a topos, providing a mathematical environment for
local-to-global interpretation of relational information.

Within this framework we established several structural results. The
line knowledge digraphs were shown to admit a decomposition determined
by equivalence classes of triples sharing a common head or tail entity,
and the line digraph constructions were analysed from a categorical
perspective through their functorial behaviour with respect to morphisms
of knowledge graphs. We further introduced a path-covering Grothendieck
topology on the free category generated by a knowledge graph and proved
that the resulting category of sheaves forms a Grothendieck topos. The
internal logic of this topos provides a natural formal setting for
context-dependent reasoning on relational data.

Several directions for further work arise naturally from this framework.
One direction concerns computational aspects of the theory, such as
algorithms for evaluating sheaf conditions or computing semantic
sections on large knowledge graphs. Finally, the framework suggests
connections with logical and conceptual formalisms used in knowledge
representation, including description logics and formal concept
analysis, which may provide additional bridges between categorical
methods and practical semantic modelling.


\section*{Declarations}

\textbf{Competing interests.} The author declares no competing interests.

\noindent \textbf{Funding.} No external funding was received for this work.


\bibliography{paper}

\end{document}